\documentclass[%
    aps, 
    physrev,
    reprint,
]{revtex4-2}

\usepackage{amsmath}
\usepackage{amssymb}
\usepackage{bm}
\usepackage{braket}
\usepackage{comment}
\usepackage[normalem]{ulem}

\usepackage{printlen}
\usepackage{float}
\usepackage{graphicx}
\usepackage{xcolor}
\usepackage[caption=false]{subfig}

\usepackage{enumitem}

\usepackage{dcolumn}
\usepackage{multirow}
\usepackage{makecell}
\usepackage{tabularx}
\usepackage{booktabs}
\usepackage{array}

\usepackage{tikz}

\usepackage[colorlinks=true, allcolors=blue]{hyperref}
\hypersetup{
    colorlinks=true,
    linkcolor=blue,
    urlcolor=blue,  
}
\usepackage[capitalise]{cleveref}
\usepackage{orcidlink}
\usepackage{quantikz}

\makeatletter
\newcommand{\vast}{\bBigg@{4}}
\newcommand{\Vast}{\bBigg@{5}}
\makeatother

\usepackage{braket}

\newtheorem{proposition}{Proposition}[section]

\newenvironment{proof}{%
    \par\noindent\textit{Proof.}\ %
}

\begin{document}

\title{Walking Floquet code circuits for zero-overhead leakage reduction}

\author{Hanna Westerheim}
\author{Kaavya Sahay}
\author{Shruti Puri}
\affiliation{Department of Applied Physics, Yale University, New Haven, Connecticut 06511, USA\\
Yale Quantum Institute, Yale University, New Haven, Connecticut 06511, USA}

\begin{abstract}

Leakage, occurring when a qubit undetectably exits the computational subspace, poses a significant challenge for quantum error correction by inducing correlated errors in space and time.
These correlations reduce both the code threshold and the effective code distance. 
To address this challenge, we introduce two zero-overhead walking circuits for the honeycomb Floquet code (hFC), termed the swirling and sliding circuits, which periodically remove leakage while dynamically protecting logical qubits via a schedule of anticommuting two-body measurements.
Unlike previous dynamic circuits for the hFC, ours preserves the distance to Pauli errors.
We further study their performance under leakage and find that, for two leakage noise models, they may correct as many leakage errors as Pauli errors,
improving on the more widely known walking surface code.  
Through numerical simulation, we observe that finite-error-rate performance under leakage is strongly influenced by entropic effects, producing a pronounced waterfall regime in which the logical error rate decreases much more rapidly with physical error rate than the expected distance-limited scaling. In this regime, even when asymptotic predictions suggest otherwise, the hFC can outperform the same-distance walking surface code in error rates and sub-threshold logical error scaling.

\end{abstract}
\date{\today}

\maketitle

\section{Introduction}

Fault-tolerant quantum computing is possible with the use of quantum error correcting codes. Quantum error correction (QEC) enables the preservation of logical information given certain types of errors on the underlying physical qubits.
However, QEC codes are particularly under-equipped to deal with \textit{leakage}, a form of error where qubits undetectably exit the computational subspace. Leakage, for example, may arise from mechanisms such as atom loss in neutral atom architectures and transmon heating in superconducting qubits~\cite{Bluvstein_2023,bluvstein2025fault, Miao_2023, Chen_2016, Dai_2026_dust, google_2023}.
A leaked qubit can cause large correlated errors across QEC codes, damaging their error suppression ability. Mitigating the impact of leakage on logical performance is becoming increasingly important as fidelities to Pauli errors improve, and leakage now accounts for a growing fraction of the error budget in several hardware platforms~\cite{computing2026quantumerrorcorrectiontoric, chen2016measure}. Consequently, QEC protocols are increasingly co-designed with additional subroutines that prevent or mitigate leakage errors. 

One such approach is the periodic reset of qubits using leakage reduction units (LRUs). LRUs can be implemented by swapping data and ancilla qubits within syndrome extraction circuits~\cite{McEwen_2023, eickbush2025demonstration}. When combined with codes that have low-weight check operators, LRUs remove leakage more rapidly, thus limiting the extent to which the resulting correlated effects can propagate through the syndrome extraction circuits~\cite{Brown_2019}.

One class of codes with low-weight check operators is the family of dynamically generated honeycomb Floquet codes~\cite{Hastings_2021,Gidney_2021}. In their static stabilizer description, they encode no logical qubits. Instead, logical information is dynamically encoded and protected through a sequence of anticommuting two-body check measurements. In this work, we introduce zero-overhead, leakage-reducing circuits for the honeycomb Floquet code (hFC). 
In contrast to previous dynamic circuits for the hFC~\cite{williamson2025dynamicalquantumcodeslogic, claes2025dynamiccircuithoneycombfloquet}, our circuits preserve the code's distance to Pauli errors and remove leakage without requiring additional gates or timesteps.
We first present the \textit{swirling Floquet code}, an intuitive implementation compatible with reconfigurable architectures.
We then introduce the \textit{sliding Floquet code}, a practical implementation designed for fixed-qubit architectures.
By leveraging the hFC's weight-2 check measurements, both implementations limit the propagation of leakage-induced correlated errors compared to the walking surface code. 
We study several physically motivated leakage noise models, demonstrating both asymptotic and practically relevant regimes where the hFC outperforms the same-distance surface code under leakage noise. 

\begin{figure*}[htbp]
    \centering
    \includegraphics[width=.85\textwidth]{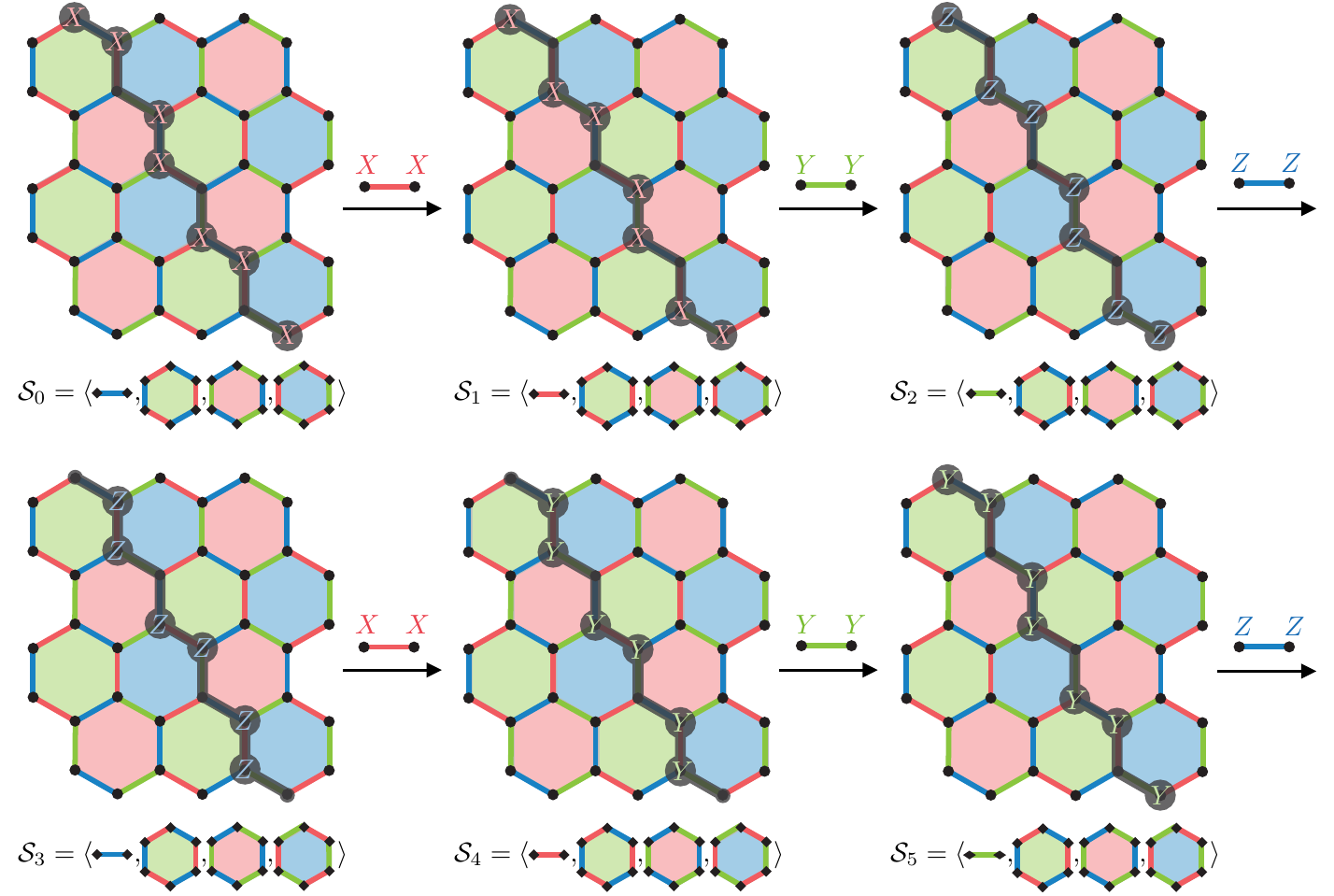}
    \caption{\textbf{The P6 honeycomb Floquet code (hFC) on the hexagonal lattice.}
    Each panel represents the hFC at a particular point in time, where we show the evolution of a logical operator representative and the ISG. Check sets represented by the edges above the arrows are measured between time steps.}
    \label{fig:logical}
\end{figure*}

This work is organized as follows. In \cref{sec:prelims}, we introduce necessary background. In \cref{sec:circs}, we present the new circuits that implement the hFC with zero overhead while periodically removing leakage. We then study their performance both analytically and numerically in \cref{sec:anadv} and \cref {sec:numres}, comparing them with the walking surface code. Finally, we conclude in \cref{sec:concl}.

\section{Preliminaries} \label{sec:prelims}

Here, we define the Floquet codes, noise models, and circuit gadgets used throughout this work. We defer a review of the leakage-reducing ``walking" surface code~\cite{McEwen_2023}, which serves as our point of comparison, to \cref{app:walksc}. 

\subsection{The honeycomb Floquet code}

The honeycomb Floquet code (hFC), shown in Figure~\ref{fig:logical}, is defined on a trivalent, three-colorable lattice forming a hexagonal tiling of a 2D plane. Each hexagonal plaquette is assigned a color $\mathbf{c} \in \{ \mathbf{r,g,b} \}$ such that no two adjacent plaquettes share the same color. An edge connecting two  $\mathbf{c}$ plaquettes is assigned the color $\mathbf{c}$. For the toric manifold with periodic boundary conditions, given a tiling with $n_p$ plaquettes, this lattice has $2 n_p$ vertices and $3 n_p$ edges.

The hFC QEC protocol on this lattice can be defined as follows. Data qubits are placed on the vertices of the toric honeycomb lattice, and subsets of the hFC's edges are measured according to a three-step periodic schedule. In this work, we consider the P6 hFC variant, in which the measured check operators at timestep $t$ are given by
\begin{equation}
\begin{aligned}
PP_\mathbf{c}(t) =
\begin{cases}
XX_\mathbf{r}, & t \bmod 3 = 0, \\
YY_\mathbf{g}, & t \bmod 3 = 1, \\
ZZ_\mathbf{b}, & t \bmod 3 = 2,
\end{cases}
\end{aligned} 
    \label{eqn:msmt schedule}
\end{equation}
where $PP_\mathbf{c}$ denotes the two-body Pauli operator $PP$ acting on edges of color $\mathbf{c}.$ This measurement sequence is illustrated in \cref{fig:logical}. On the toric manifold, this procedure effectively realizes a $[[2n_p,2,\sqrt{n_p}]]$ dynamical QEC code.

Note that consecutive $XX_\mathbf{r} \rightarrow YY_\mathbf{g} \rightarrow ZZ_\mathbf{b}$ measurements incident to same vertex anticommute. 
As a result, instead of using a static stabilizer group ubiquitous to stabilizer codes, an \textit{instantaneous stabilizer group} (ISG) is used for error correction.
The ISG for the P6 code, represented by plaquettes in \cref{fig:logical}, is generated by six-body operators which survive the individually anticommuting check measurements, as well as an evolving set of two-body check operators.
In the periodic steady state, the ISG is given by:
\begin{equation}
    \begin{aligned}
\mathcal{S}_t = \{ X^{\otimes 6}_{ \mathbf{r}}, Y^{\otimes 6}_{\mathbf{g}} ,
Z^{\otimes 6}_{ \mathbf{b}} \} \cup
\begin{cases}
XX_r, & t \bmod 3 = 0, \\
YY_g, & t \bmod 3 = 1, \\
ZZ_b, & t \bmod 3 = 2,
\end{cases}
\end{aligned}
    \label{eqn:isg}
\end{equation}
where, by abuse of notation, $P^{\otimes 6}_{ \mathbf{c}}$ denotes the  set of 6-body Pauli operators supported on all plaquettes of color~$\mathbf{c}$. 
In \cref{fig:logical}, the ISG at each timestep $t$ is denoted by $\mathcal{S}_t$ below each corresponding panel.

\begin{figure*}[t]
\centering
\includegraphics[width=\textwidth]{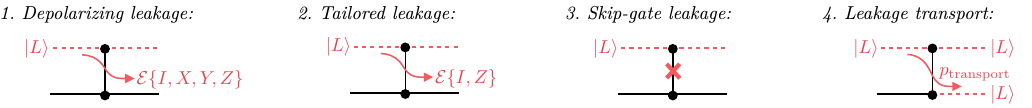}
\caption{\textbf{The four leakage models considered in this work}. The models differ in the effect of a leaked qubit on a sealed qubit upon interaction through a two-qubit Clifford CP gate, where P is a Pauli (CZ is shown as an example).} \label{fig:leakage-models}
\end{figure*}

Error correction in the hFC is performed using the values of the above six-body plaquette operators, each given by the product of the six edges along the boundary of the plaquette. We defer further details of error correction and decoding to  \cref{app:detectors}. 

We note that other Floquet code variants have been proposed, differing in their lattice geometries~\cite{alam2025dynamical}, measurement schedules~\cite{Davydova_2023}, and resilience to biased noise~\cite{Setiawan_2025}. 
The results of this work, though focused on the P6 hFC, can be extended in a straightforward manner to these alternate variants.

\subsection{Leakage noise models}

Throughout this work, we employ a uniform leakage noise model in which, at every circuit spacetime location, a qubit may leak from the computational subspace with probability $p$. The effect leakage has on logical performance depends on 
what happens when a leaked qubit interacts with another
un-leaked, i.e. {sealed}, qubit through a controlled-Pauli CP gate such as CX or CZ. In this paper, we consider the following four experimentally motivated leakage effects, summarized in \cref{fig:leakage-models}.

\textit{1. Depolarizing leakage:} The sealed qubit is completely depolarized by the channel $\mathcal{E}\{I,X,Y,Z\}$~\cite{chang2024surfacecodeimperfecterasure,pavlovich2026erasuresurfacecodecircuit}, i.e.,
\begin{equation}
    \rho \rightarrow \frac{\rho +
    X \rho X + 
    Y \rho Y + Z \rho Z 
    }{4} .
\end{equation}

\textit{2. Tailored leakage:} The sealed qubit experiences either the channel $\mathcal{E}\{I,Z\}$ or $\mathcal{E}\{I,P\}$, depending on whether it was the control or target of the CP gate, respectively. We have
\begin{align}
  \mathcal{E}\{I,Z\} &: \rho \rightarrow \frac{\rho + Z \rho Z }{2} ,\\
  \mathcal{E}\{I,P\} &: \rho \rightarrow \frac{\rho + P \rho P }{2}.
\end{align}

\textit{3. Skip-gate leakage:}
The action on the sealed qubit is as if the CP gate did not occur, i.e.,
$$ \rho_{\text{leaked}} \otimes \rho
\xrightarrow {CP}
\rho_{\text{leaked}} \otimes \rho .$$

\textit{4. Leakage transport:} It is possible for leaked qubits to enhance the probability of the sealed qubits they interact with to leak ~\cite{Miao_2023}. To understand the effect of such leakage transport we will consider a noise channel where a sealed qubit may leak with probability $p_{\mathrm{transport}}$ if it interacts with a leaked qubit:
\begin{equation}
    \rho \rightarrow (1-p_{\mathrm{transport}}) \rho + p_{\mathrm{transport}} \ket{L}\bra{L} .
\end{equation}

\subsection{Logical error rate ans\"atze}

In general, the logical error rate of QEC codes may be expressed as
\begin{equation}
    p_{L}
    =
    \sum_{w=w_{\min}}^{N}
    N_w p^w (1-p)^{N-w},
\end{equation}
where $N$ is the total number of circuit fault locations, $N_w$ is the number of malignant error configurations of weight $w$, and $w_{\min}$ is the minimum weight of any malignant error configuration. Here, we refer to an error configuration that causes logical failure as \emph{malignant}.

In the asymptotic low-error-rate regime, the logical error rate is dominated by minimum-weight malignant configurations, yielding
\begin{equation}
    p_{L}
    =
    N_{w_{\min}} p^{w_{\min}}
    + \mathcal{O}\!\left(p^{w_{\min}+1}\right)
    \sim p^{w_{\min}},
\end{equation}
where $w_{\min}$ depends on the code, decoder, and error model. However, this asymptotic regime often proves difficult to probe~\cite{beverland2025failfasttechniquesprobe}, and may not be relevant for quantum hardware with component error rates $\gg 10^{-5}$.
At these experimentally relevant physical error rates, higher-weight malignant configurations can contribute substantially due to their combinatorial multiplicity. In this case the logical error rate is dominated by these entropic contributions, a regime we refer to as the \emph{waterfall regime}, while only at sufficiently low error rates may the asymptotic scaling emerge~\cite{english2025isingdonutregimestopological,gu2026scalableneuraldecoderspractical}. We find that for the physical error rates simulated in this paper the logical error rate fits very well to a single exponential $\sim p^{\alpha}$, and we refer to $\alpha$ as the effective scaling exponent.

\subsection{Leakage mitigation strategies}

\textit{Leakage reduction units} are circuit gadgets designed to remove leakage via the transfer of  information onto freshly initialized physical qubits. Since these gadgets may themselves be faulty, it is desirable that LRUs are compact and low-overhead. In this work, we consider LRUs that can be folded into syndrome extraction circuits (SECs) with no additional qubits, gates, or circuit depth.

We begin with the example of a SEC \textit{without} leakage reduction. In the hFC, syndrome extraction is performed by measuring the two-body check operator associated with each lattice edge. The circuit in~\cref{fig:standard-sec} measures a $ZZ_b$ check, corresponding to a $ZZ$ operator acting on the pair of data qubits joined by a blue edge in the lattice.
Measuring the ancilla yields the $ZZ$ eigenvalue while projecting the data qubits onto the corresponding $ZZ$ eigenspace.

\begin{figure}[ht]
\centering
\includegraphics[width=\columnwidth]{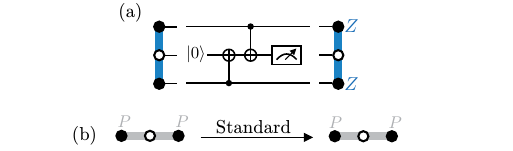}
\caption{\textbf{Standard syndrome extraction circuit (SEC).} (a) The ancilla qubit (middle wire) is reset, entangled with the two data qubits, and then measured. The measurement yields the eigenvalue of $ZZ$ on the data qubits. (b) The standard SEC preserves data/ancilla qubit role assignments, unlike the late- and early-SWAP SECs shown in \cref{fig:late-swap,fig:early-swap}.}
\label{fig:standard-sec}
\end{figure}

LRU gadgets can be integrated into SECs to remove leakage. In this work, we employ SWAP gates which transfer data qubit information onto freshly initialized ancillae. By periodically swapping the roles of data and ancilla qubits throughout the QEC cycle, \textit{all} physical qubits are periodically reset.
This is in contrast to the standard approach, where only a static subset of physical qubits which serve as ancillae are periodically reset.
As previously stated, it is ideal to perform these SWAP operations with no additional overhead relative to the standard SEC. To achieve this, we make use of the following basic circuit identities in~\cref{fig:identities}.
\begin{figure}[ht]
\centering
\includegraphics[width=\columnwidth]{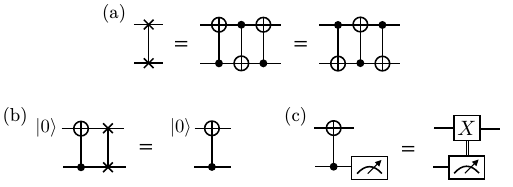}
\caption{\textbf{SWAP gate identities.} (a) The SWAP gate can be decomposed into three CX gates. (b) The SWAP gate on the left-hand side of the equality is trivial.
(c) A CX gate followed by measurement in the $Z$($X$) basis on the control(target) qubit can be compiled as a classically controlled $X$($Z$) gate.}
\label{fig:identities}
\end{figure}

By applying these identities, a zero-overhead SWAP may be realized at various points within the SEC. We employ two distinct SEC-LRU gadgets: the \textit{late-SWAP} and the \textit{early-SWAP} circuits. In each of these cases, the resulting condensed circuits simultaneously (1) read out the syndrome while leaving the remaining qubits in the corresponding eigenspace, and (2) swap information between ancillae and data qubits without any additional qubits, gates, or circuit depth when compared to the standard SEC.

\textit{1. Late-SWAP:} In this case, a SWAP is added after the final gate of the SEC, exchanging the roles of the ancilla and a data qubit \cite{Suchara2015,Brown2020CriticalFaults,Brown2019MixedQubit}, illustrated below in~\cref{fig:late-swap}. The reduced circuit on the right~\cite{Baranes_2026leveraging,shaw2026optimisingquantumerrorcorrection,kim2026timedynamiccircuitsfaulttolerantshift} utilizes the identities in ~\cref{fig:identities}(a) and (c). In practice, the corresponding classical update can be absorbed into the decoder.
In \cref{fig:late-swap}, observe that the $ZZ$ operator information moves from the top and bottom qubits to the bottom two qubits.

\begin{figure}[ht]
\centering
\includegraphics[width=\columnwidth]{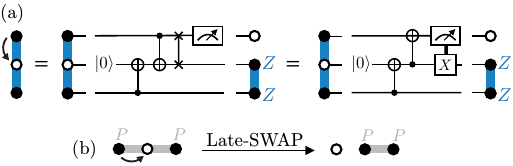}
\caption{\textbf{Late-SWAP syndrome extraction circuit.} (a) The central term is a schematic representation of the late-SWAP circuit, where just before measurement, a SWAP is applied between the ancilla qubit and one of the data qubits. The right-hand side employs~\cref{fig:identities}(a) and (c), yielding the final, zero-overhead circuit. (b) The late-SWAP SEC exchanges the roles of one of the data qubits and the ancilla, with the solid arrow indicating the direction of the data qubit information transfer.}
\label{fig:late-swap}
\end{figure}

\textit{2. Early-SWAP:} In this case, the SWAP gate is inserted immediately after the first gate~\cite{liu2026achievingoptimaldistanceatomlosscorrection, pavlovich2026erasuresurfacecodecircuit}, shown in \cref{fig:early-swap}.
Here, the circuit identity in~\cref{fig:identities}(b) can be used to remove the extra gates associated with the SWAP in order to produce the reduced circuit on the right.

\begin{figure}[H]
\centering
\includegraphics[width=\columnwidth]{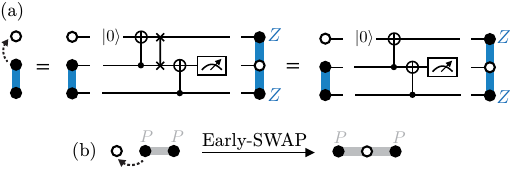}
\caption{\textbf{Early-SWAP syndrome extraction circuit.} (a) The central term is a schematic representation of the early-SWAP circuit implementation, where a SWAP is applied between the ancilla qubit and one of the data qubits just after the first CX gate. The right-hand side employs~\cref{fig:identities}(b), yielding the final, zero-overhead circuit. (b) The early-SWAP SEC transfers information from one data qubit to the freshly initialized ancilla, as indicated by the dashed arrow.}
\label{fig:early-swap}
\end{figure}

Finally, observe that neither of these LRUs increase the qubit count, gate count, or depth of the SEC, as desired, and are thus zero-overhead.
In previous literature, circuits have implicitly used either late-SWAP or early-SWAP circuits. In this work, we show how both can be integrated into one circuit for the first time in order to achieve a compact, hardware-amenable implementation.

\begin{figure*}[t]
    \centering
    \includegraphics[width=\textwidth]{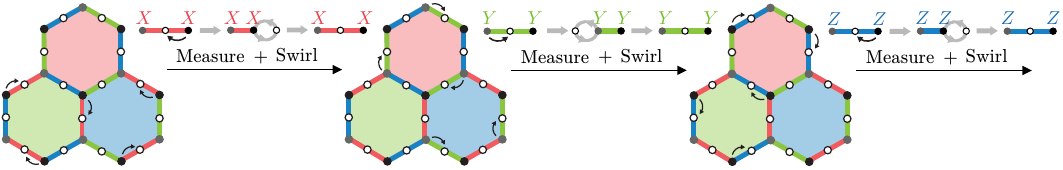}
    \caption{\textbf{Schedule of late-SWAP LRUs in the swirling hFC.} Solid curved arrows indicate late-SWAP LRUs, with arrowheads indicating the direction of data qubit information flow. Black (gray) data qubits are swapped with white ancillae during even (odd) subrounds, respectively, resulting in an intermediate configuration of data and ancilla qubits immediately after measurement. Physical rearrangement (``swirling") then returns the data and ancilla qubits to their original locations before the next measurement subround.}
    \label{fig:swirl}
\end{figure*}

\section{Zero-overhead hFC LRU circuits} \label{sec:circs}

The late- and early-SWAP LRU gadgets provide a clean framework to create zero-overhead leakage-reducing circuits for the hFC. We lay out two distinct circuits that leverage these gadgets in different ways. The first, the \textit{swirling} hFC, is an intuitive realization using exclusively late-SWAP LRUs. The swirling hFC achieves zero-overhead leakage reduction in reconfigurable neutral atom architectures, integrating a local swirling rearrangement of qubits between successive check measurement layers. The second, the \textit{sliding} hFC, is designed for fixed-qubit architectures and employs both early- and late-SWAP LRUs. 
Unlike previous leakage-reducing hFC proposals~\cite{williamson2025dynamicalquantumcodeslogic, claes2025dynamiccircuithoneycombfloquet}, both circuits natively preserve the timelike and spacelike code distances under Pauli noise. We provide these fault tolerance proofs for the late- and early-SWAP SECs for the hFC in~\cref{app:fault-tol}, employing the methods presented in Ref.~\cite{rodatz2026faulttoleranceconstruction}.

\subsection{Swirling hFC }

In the swirling hFC, we apply a late-SWAP LRU gadget during the measurement of every check operator. 
Since each check acts on two data qubits, we can choose one of the data qubits each measurement round to be swapped. To maintain a uniform schedule, we assign a bipartite coloring (black and gray in \cref{fig:swirl}) to the data qubits and choose to reset the black (gray) qubits during even (odd) rounds.

The late-SWAPs used to perform syndrome extraction in the swirling hFC are indicated by the solid black arrows within the code patches in \cref{fig:swirl}, which point in the direction that data qubit information moves. In between rounds in \cref{fig:swirl}, we zoom in on one edge and show a close-up view of the ``measurement + swirl" procedure used in the this hFC variant. During the measurement step, information from one of the data qubits is transferred to the ancilla position via the late-SWAP SEC. At this point, the physical qubit at the center of the edge is now in the role of the data qubit. Next, in order to enable the subsequent round of check measurements, we conceptually envision physically transporting these data qubits back to their original locations, giving rise to a local \textit{swirling} motion.
This description is primarily pedagogical: in practice, in a neutral-atom setting, the swirling motion may be incorporated into the atom movement required for the subsequent two-qubit gate, rather than occurring as an independent transport step.

In \cref{app:allhfcs}, in addition to a zero-overhead ``pipelined" swirling hFC variant, we describe a hFC circuit implementation designed for fixed heavy-hex connectivity, where this physical motion is forgone at the cost of an extra teleportation gate.

\begin{figure*}[t]
    \centering
    \includegraphics[width=.9\textwidth]{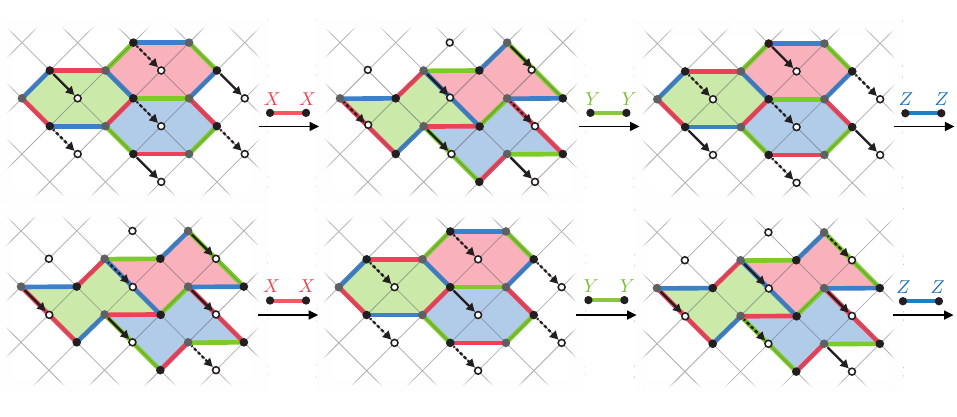}

    \caption{\textbf{Schedule of early- and late-SWAP LRUs in the sliding hFC}, which is compatible with fixed-qubit architectures with square-connectivity. In each panel, we show the LRU arrangement for the indicated upcoming $PP_\mathbf{c}$ measurement, with the post-measurement patch configuration shown in the following panel. Solid arrows in each panel indicate late-SWAP LRUs, while dashed arrows indicate early-SWAP LRUs. The choice of LRU circuit is assigned according to the available local connectivity of data and ancilla qubits, with arrows indicating the direction of data qubit information transfer. In even (odd) subrounds, black (gray) data qubits are swapped.
    Note that in any given panel, colored edges only represent the support of  $PP_\mathbf{c}$ operators, i.e., the square connectivity (gray grid) does not support connections between the qubits at the endpoints of horizontal edges.}

    \label{fig:sliding}
\end{figure*}

\subsection{Sliding hFC}

The swirling hFC is not compatible with fixed-qubit architectures. We now present the \textit{sliding} hFC, which is compatible with fixed nearest-neighbor connectivity. In the sliding hFC, qubits are placed on the vertices of a square lattice, as shown in \cref{fig:sliding}. Unlike previous arrangements for embedding a honeycomb lattice onto square connectivity~\cite{gidney2023newcircuitsopensource}, each hexagonal plaquette is centered around a single lattice site.

As in the swirling hFC, we alternate the globally labeled data qubits that are swapped: in \cref{fig:sliding}, black qubits are swapped during even subrounds, while gray qubits are swapped during odd subrounds.
Unlike the swirling hFC, however, the sliding hFC uses both late- (solid arrows) and early-SWAP (dashed arrows) circuits to realize the desired motion while maintaining square-lattice nearest-neighbor connectivity.

The arrangement of early- and late-SWAPs follows from the internal structures of the two LRU circuits, resulting in different connectivity requirements for fixed architectures.
As implied by \cref{fig:late-swap}, the late-SWAP circuit requires the ancilla to be adjacent to both data qubits involved in the measurement in a fixed layout. 
In contrast, the early-SWAP circuit (\cref{fig:early-swap}) requires the ancilla to be adjacent to the data qubit involved in the swap. This data qubit must then be adjacent to the other data qubit involved in the check measurement.

\begin{figure}
\centering
\includegraphics[width=\columnwidth]{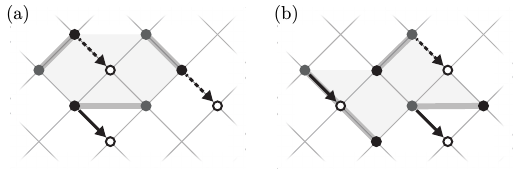}
\caption{\textbf{Local implementation of the sliding hFC} during (a) even and (b) odd check subrounds. Gray edges indicate the edges measured in the corresponding subround, while solid and dashed arrows denote late-SWAP and early-SWAP circuits, respectively. Arrows point in the direction of data-qubit motion. In (a), black data qubits are swapped with white ancillae, while in (b), gray data qubits are swapped. }
\label{fig:sliding-close-up}
\end{figure}

Under these constraints, we therefore assign the early- and late-SWAP circuits to edges based on their available qubit connectivity, as illustrated in \cref{fig:sliding-close-up}. 
In \cref{fig:sliding-close-up}(a), the highlighted gray edges are measured during an even subround, swapping the black data qubits. Solid arrows indicate where a late-SWAP circuit is applied, while dashed arrows indicate application of the early-SWAP circuit. 
The resulting data qubit configuration is displayed in \cref{fig:sliding-close-up}(b), showing an effective distortion of the lattice.

Crucially, this distorted lattice takes a form such that a subsequent round of carefully arranged LRUs restores its original shape.
\cref{fig:sliding-close-up}(b) displays this arrangement of LRUs applied at an odd subround: the complimentary set of edges is measured, the gray data qubits are swapped, and the early- and late-SWAP circuits are assigned according to the available connectivity. 
As desired, the second series of SWAPs restores the hexagonal lattice configuration up to a diagonal translation.
In both panels, the arrows indicate the direction in which the data qubit information flows.

Applying the above strategy globally at each round gives us the period-six procedure shown in \cref{fig:sliding}. Every round, half of the data qubits move diagonally down and to the right by one unit, and then the other half during the next round. The result is that at every other round, the lattice is restored, but translated down and to the right by one unit, hence the name ``sliding."
Note that, by repeated application of this LRU configuration, subsequent sliding steps will continue to move the lattice in the same direction. One may choose to reverse the direction of lattice displacement in order to have the lattice ``wiggle" back and forth. Here, we do not make this choice, for reasons we further explain in \cref{app:counting}.

Finally, we comment that the sliding and swirling hFCs differ in their total qubit requirements.
Note that there are $2d^2$ data qubits. 
In the swirling hFC, all $3d^2$ edges of the lattice are equipped with ancilla qubits, giving a total of $5d^2$ physical qubits. In the sliding hFC, ancilla are reused between subrounds. Thus, the sliding hFC requires only $d^2$ ancilla qubits, for a total of $3d^2$ physical qubits.

\section{Analytical hFC advantages} \label{sec:anadv}

We now analytically study the advantages of the swirling and sliding hFCs over the walking surface code. We show that these hFC circuits may offer advantages in terms of qubit lifetime and asymptotic error scaling.

\subsection{Qubit lifetimes} \label{sec:qublifetimes}

One way we can quantify the effects of leakage is by studying the ``lifetime" of each physical qubit, i.e., the number of entangling gates it experiences between initialization and measurement. Here, a shorter lifetime is desirable, as it limits the number of subsequent gates over which a leakage error can persist and propagate errors.

\begin{figure}
\centering
\includegraphics[width=\columnwidth]{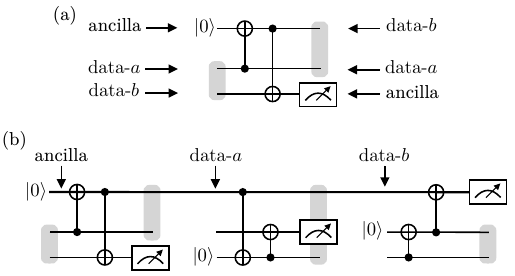}
\caption{\textbf{Counting the qubit lifetime for the swirling hFC.}
(a) Late-SWAP circuit for a $ZZ$-edge measurement. Data qubits are labeled $a$ and $b$, with the ancilla swapping with the data-$b$ qubit during the circuit.
(b) A qubit (top wire) is tracked from initialization to measurement as it changes roles in the late-SWAP circuit. The qubit is initialized as an ancilla and subsequently becomes a data qubit. 
Specifically, in the next round, it takes the role of data-$a$ and is not involved in a SWAP. For simplicity, we again show a $ZZ$-edge measurement, although the next round measures a different edge type in the actual schedule. In the final round, the qubit takes the role of data-$b$, swaps with the ancilla, and is subsequently measured out. Thus, the qubit experiences four gates between initialization and measurement.}
\label{fig:swirling-counting}
\end{figure}

For the walking surface code, we see that qubits consistently experience a lifetime of $8$ gates, shown in \cref{app:walksc}.
In comparison, for the swirling hFC, we find that qubits have a lifetime of only $4$ gates. Figure~\ref{fig:swirling-counting} illustrates this counting for the swirling hFC. 
For the sliding hFC, the lifetime is not uniform, but averages to $4$ gates, which we detail in Appendix~\ref{app:counting}.
Thus, relative to the walking surface code, our hFC circuits reduce the qubit lifetime on average by a factor of two. This reduction limits the number of subsequent operations over which a leakage error can persist, providing an intrinsic advantage of the hFC circuits for mitigating the effects of leakage.

\begin{figure*}[ht]
    \centering
    \includegraphics[width=\textwidth]{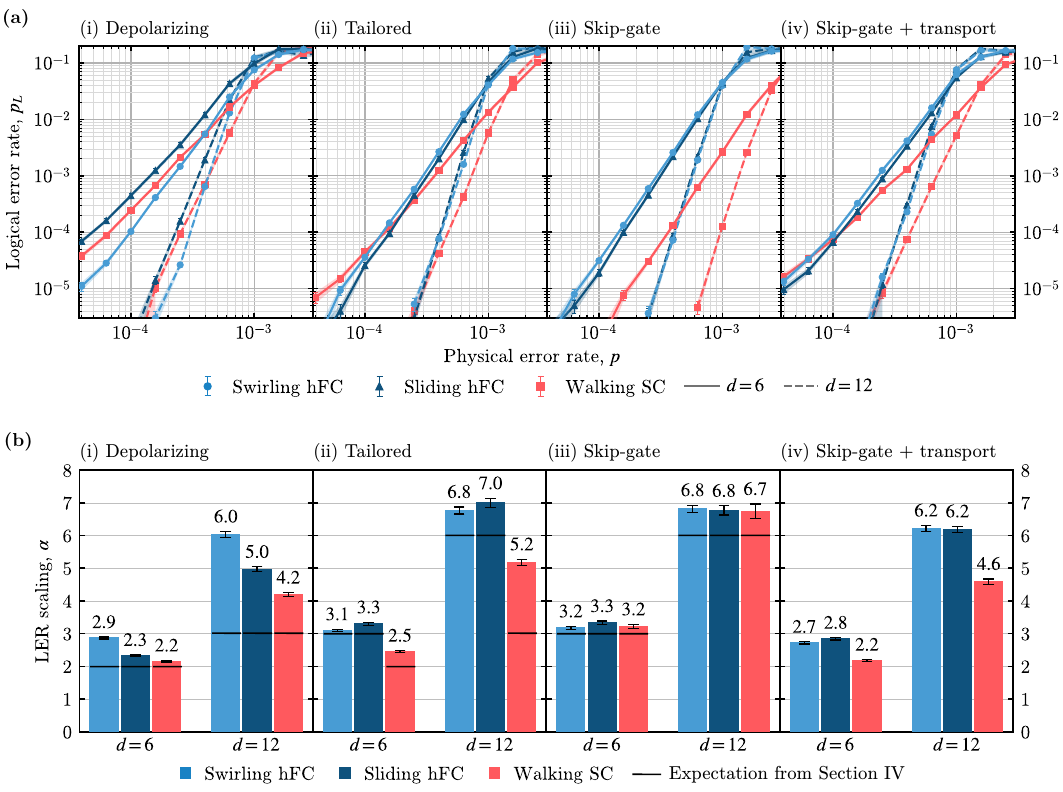}
    \caption{\textbf{Numerical comparison of the swirling and sliding hFCs with the walking surface code} under (i) depolarizing leakage, (ii) tailored leakage, (iii) skip-gate leakage, and (iv) skip-gate leakage with $10\%$ leakage transport.
    Top panels in (a) show the LER per $d$ rounds of error correction ($3d$ QEC rounds simulated), where one QEC round corresponds to three check subrounds for the hFC. The bottom panels in (b) show the scaling exponents $\alpha$ of LERs fitted to $\sim p^\alpha$.}
    \label{fig:bars}
\end{figure*}

\subsection{Minimum weight errors}

As discussed earlier, the asymptotic LER is governed by the minimum malignant fault weight. 
These fault weights have been benchmarked for the walking surface code under leakage noise in Ref.~\cite{pavlovich2026erasuresurfacecodecircuit}.
Thus, to establish a comparison, we aim to determine this value for the swirling and sliding hFC variants. 

Ideally, one would determine the minimum malignant fault weight $w_{\min}$ by exhaustively searching over all possible fault configurations.
The lowest-distance hFC that we may fairly benchmark is $d=6$, as it is the smallest code on a torus providing equal protection to both logicals. 
However, even for $d=6$, such a brute-force search is computationally prohibitive under leakage noise.
In particular, because faults can occur throughout the three-dimensional spacetime volume, the search space grows rapidly with the number of quantum error correction rounds. 

We therefore instead construct adversarial fault configurations and increase the weight until we produce a logical error~\footnote{We construct adversarial leakage fault configurations by placing leakage events along a minimum length homologically nontrivial loop, and incrementally add leakage events to this loop until a decoding failure occurs. 
The resulting fault weight provides an upper bound on the minimum malignant fault weight, rather than a rigorous determination of it, since lower-weight malignant configurations cannot be ruled out. 
We use this bound as a prediction to aid in interpreting the numerical results in the finite, experimentally relevant error rate regime.}. 
The resulting fault weights provide upper bounds on the minimum malignant fault weight $w_{\min}$, and we use these bounds to motivate predictions for the asymptotic LER scaling (summarized in \cref{tab:asymptotic_scaling}).

We observe that the walking hFC variants match the walking SC asymptotic scaling in terms of expected least-weight scaling for depolarizing and skip-gate noise. Further, the hFC variants are expected to double the SC asymptotic scaling under tailored leakage noise, implying that, for the same number of physical qubits, the walking hFC may exhibit asymptotically favorable LER scaling over the walking SC.

We emphasize that these bounds are not rigorous proofs of the minimum weight, as lower-weight malignant configurations cannot be ruled out by this procedure.
Nevertheless, as we will show in \cref{sec:numres}, the numerically observed LER scaling exponents under practical physical error rates exceed even these upper bounds on $w_{\min}$. This difference reflects the contribution of higher-weight malignant configurations and their combinatorial multiplicity in the experimentally relevant waterfall regime.

\begin{table}[t]
    \centering
    \caption{\textbf{Upper bounds on the minimum malignant fault weight.} Weights of adversarially planted leakage fault configurations that produce logical errors for the walking hFC (both  swirling and sliding) and the walking surface code under each leakage model. These weights provide upper bounds on $w_{\min}$ and serve as predictions for the asymptotic LER scaling; they are not proven minimum-weight values.}
    \label{tab:asymptotic_scaling}

    \begin{tabular*}{\columnwidth}{@{\extracolsep{\fill}}lcc}
        \toprule
        \textbf{Noise Model}
        & \shortstack{\textbf{Walking hFC}\\\footnotesize($n = 3d^2$ or\footnote{Sliding or swirling hFC variants, respectively.} $5d^2$)}
        & \shortstack{\textbf{Walking SC}\\\footnotesize($n = 2d^2$)} \\
        \midrule
        Depolarizing
        & $\lceil d/4 \rceil$
        & $\lceil d/4 \rceil$ \\
        Tailored
        & $\lceil d/2 \rceil$
        & $\lceil d/4 \rceil$ \\
        Skip-gate
        & $\lceil d/2 \rceil$
        & $\lceil d/2 \rceil$ \\
        \bottomrule
    \end{tabular*}
\end{table}

\section{Numerical results} \label{sec:numres}

We now numerically compare the swirling and sliding hFC circuits with the walking surface code under the noise models introduced in \cref{fig:leakage-models}. We evaluate the logical error rates (LERs), shown in \cref{fig:bars}(a), and extract their scaling with physical error rate under circuit level noise by fitting the LERs to $\sim p^\alpha$ for each code, distance, and leakage noise model, with the resulting scaling exponents $\alpha$ summarized in \cref{fig:bars}(b). 
For the surface code, we report the LER of its single logical qubit, while for the hFC, we report the LER of one of its two logical qubits, consistent with prior work~\cite{Gidney_2021}. 
See \cref{app-sims} for details on the implementation of these simulations.
These results allow us to assess the finite-error-rate performance of the hFC circuits and compare the observed scaling with the asymptotic predictions of \cref{sec:anadv}.

Let us discuss these results by noise model, beginning with depolarizing noise, shown in \cref{fig:bars}(a)(i) and \cref{fig:bars}(b)(i). As the most unstructured noise model, we showed in the previous section that the asymptotic distance scaling is reduced for both the hFC and the SC in this case (relative to the $\lceil d/2 \rceil$ LER scaling expected under Pauli noise).
However, at practically relevant values of $p$ we observe a notable deviation from this asymptotic behavior: the LER of the swirling hFC (light blue in \cref{fig:bars}) exhibits a steeper scaling with $p$ than that of the SC (red). This indicates that, in this regime, higher-weight malignant fault configurations contribute more significantly to the LER of the swirling hFC than to that of the SC. The sliding hFC (dark blue) likewise exhibits a steeper scaling than the SC, although the improvement is less pronounced than for the swirling variant. 
We hypothesize that this difference may be related to the longer maximum leakage lifetime of the sliding hFC compared with the swirling hFC (see \cref{app:counting}), which could lead to a different distribution of malignant fault configurations.

Let us further examine \cref{fig:bars}(a)(i) by moving from high to low $p$, motivated by progression of experimental improvements in leakage rates over time. 
At high $p$, the larger number of malignant fault configurations in the hFC results in a higher LER than that of the SC, and recent experimental demonstrations may fall in this regime.
As $p$ decreases, however, we observe a crossover below which the swirling hFC outperforms the SC, despite the absence of a corresponding advantage in the asymptotic scaling. 
This improvement arises from the greater contribution of higher-weight fault configurations in the swirling hFC, which leads to a steeper LER scaling in the experimentally relevant waterfall regime.
Importantly, continued improvements in physical error rates may push hardware into this crossover regime.
We expect this behavior to persist at even lower $p$, beyond the range accessible to our simulations, until a point at which both codes recover their asymptotic scaling. 
Some platforms may never reach this asymptotic regime, highlighting the importance of the waterfall regime for practical error correction.

For the tailored noise model shown in \cref{fig:bars}(a)(ii) and \cref{fig:bars}(b)(ii), we again observe LER behavior that differs from the asymptotic predictions.
Recall that our analysis in \cref{sec:anadv} predicts that the hFC has a higher minimum malignant fault weight than the SC.
This would lead us to expect greatly improved LERs and scaling for the hFC compared to the SC under tailored noise.
However, in practice, we find that the observed LER scaling of the SC is only slightly weaker than that of both the swirling and sliding hFCs in the practically relevant regime.
We attribute this notable discrepancy in the SC's numerical scaling to entropic factors. 
At the error rates considered here, the LER of the SC is dominated by contributions from higher-weight errors. 
Consequently, although we do observe a crossover below which the hFC outperforms the SC, this crossover occurs at a lower $p$ than for depolarizing noise.

The results for the skip-gate model, shown in the \cref{fig:bars}(a)(iii) and \cref{fig:bars}(b)(iii), are also illuminating. 
For this highly structured noise model, all codes exhibit little deviation from the expected $\lceil d/2 \rceil$ scaling.
Thus, as expected when we do not have leakage transport, the SC outperforms the hFC due to a fewer number of weight-$\lceil d/2 \rceil$ error mechanisms that produce logical failure.

For purely skip-gate leakage, we therefore do not expect a crossover between the two codes at which the hFC outperforms the SC.
However this conclusion changes markedly in \cref{fig:bars}(a)(iv) and \cref{fig:bars}(b)(iv), when leakage transport is introduced. 
With $10\%$ leakage transport, both codes exhibit reduced LER scaling exponents under the skip-gate noise model, but the degradation is more pronounced for the SC than for the hFC.
Evidently, weight $<\lceil d/2 \rceil$  fault configurations contribute to logical failure in the presence of leakage transport. 
Notably, these configurations affect the SC more than the hFC.
As a result, we again observe a crossover $p$ below which the hFC outperforms the SC.

Across all noise models considered in this work, we find that entropic effects become more pronounced at larger code distances. In particular, the observed scaling at $d=12$ generally deviates more strongly from the expected asymptotic scaling than at $d=6$, as highlighted in \cref{fig:bars}(b), with the effect most clearly observed across all circuits under depolarizing leakage noise. This underscores that finite-error-rate entropic effects can become increasingly important as code distance grows, making asymptotic scaling alone insufficient to characterize the practical performance of these codes.

Overall, these results demonstrate that the practical performance of the hFC relative to the SC is strongly dependent on the structure of the underlying noise. 
Our asymptotic analysis in \cref{sec:anadv} provides upper bounds on the minimum malignant fault weight, yet the observed finite-error-rate scaling meets or greatly exceeds these bounds in every case we study due to contributions from higher-weight fault configurations and their combinatorial multiplicity. 
In particular, these entropic effects can produce regimes in which the hFC outperforms the SC even in cases where it has no predicted asymptotic advantage, highlighting the importance of finite-error-rate behavior when comparing QEC protocols.

\section{Conclusion} \label{sec:concl}

Through the design and analysis of two novel circuits, we have established the advantages the hFC presents over the surface code under leakage noise. 
Specifically, we show that our walking hFC circuits, using zero-overhead leakage reduction gadgets, reset leakage more frequently than the walking SC.
Further, the hFC is subject to a steeper logical error rate scaling under a variety of experimentally-motivated leakage models, resulting in regimes where it outperforms the same-distance SC.
Our discoveries complement Refs.~\cite{dessertaine2026enhancedfaulttolerancephotonicquantum, Gidney_2022}, rounding out previous works that establish regimes where the hFC forms the code of choice over the surface code.

We now begin with forward-looking comments specific to our work. Our analysis is currently conducted for the toric honeycomb lattice.
For practical implementation, it would be desirable to benchmark the corresponding planar code variant with boundaries. 
We note that, since stabilizers at code boundaries are generally lower weight, we expect the scaling to remain the same or improve for these lattices. Additionally, our work and the community at large would greatly benefit from more accurate simulators for low error rates and better LER ans\"atze for leakage noise, for which the splitting method~\cite{bravyi_pra2013, mayer2025rareeventsimulationquantum} cannot be directly used.
Next, it would be interesting to analyze the gain in performance for these circuits when using information from erasure checks or three-state readout for leakage detection~\cite{gu2025erasure}. 
Any potential advantages may be dependent on new accurate decoding strategies to optimally use the additional measurement information ~\cite{liu2026achievingoptimaldistanceatomlosscorrection, pavlovich2026erasuresurfacecodecircuit, Gu_2025}.

We have observed that the weight two check measurements indeed mitigate the spread of leakage. This brings about the question as to whether there exists a dynamical measurement schedule for a subsystem code with low-weight checks that can outperform the walking surface code in even more regimes. For instance, one may consider designing zero-overhead LRU circuits for newer constructions such as that of Refs.~\cite{zen2026lowvalencyscalablequantumerror,alam2026baconshorboardgames}.

More broadly, our results motivate designing leakage-resilient QEC circuits while jointly optimizing circuit structure, leakage effects, connectivity constraints and decoding. Mapping out the Pareto-frontier of this parameter space will help identify favorable architectures under different hardware constraints, and ultimately guide the design of next-generation quantum hardware and fault-tolerant protocols.

\textit{Acknowledgments --} We are grateful to Margaret Pavlovich, Pengyu Liu, Katie Chang, John Garmon, and Benjamin Rodatz for useful discussions. 
We thank the Yale Center for Research Computing for use of the Grace and Bouchet clusters. This work was supported by the Army Research Office (W911NF2410358) and by the National Science Foundation (NSF) Graduate Research Fellowship Program under Grant No. DGE-2139841. Any opinions, findings, and conclusions or recommendations expressed in this publication are those of the authors and do not necessarily reflect the views of NSF. 
\bibliography{references.bib}

\appendix
 \crefalias{section}{appendix} 

\setcounter{equation}{0}
\setcounter{figure}{0}
\setcounter{table}{0}

\renewcommand{\thefigure}{A\arabic{figure}}
\renewcommand{\thetable}{A\arabic{table}}

\newpage

\section{Walking surface code}\label{app:walksc}
In this work, we compare our hFC circuits against the walking surface code, implemented  exactly as in Ref.~\cite{McEwen_2023}. Specifically, we compare against the ``wiggling" variant, which we describe in this section.

The walking surface code can be understood in terms of the LRU circuits introduced in the main text.
In particular, a late-SWAP LRU gadget is applied at the end of each syndrome extraction round. This is illustrated in \cref{fig:walking-sc}(left), where late-SWAPs are represented by solid black arrows indicating the direction of data qubit information flow.
In order to restore the original lattice locations, the direction of the late-SWAPs is reversed every round, as in \cref{fig:walking-sc}(right).

The walking surface code has a qubit lifetime of $8$ gates. To see this, note that the syndrome extraction cycle of the conventional surface code has depth $4$. The addition of a late-SWAP does not increase this depth, and maps data qubit information onto a freshly initialized ancilla qubit at the end of a syndrome extraction round. This qubit then participates in the four gates of the subsequent syndrome extraction round before being measured out. Thus each qubit undergoes $4$ gates in the round in which it acts as an ancilla and $4$ gates in the following round, in which it acts as a data qubit, resulting in a lifetime of $8$ gates. This pattern is uniform across the lattice, up to boundary effects.

\begin{figure}
    \centering
    \includegraphics[width=\linewidth]{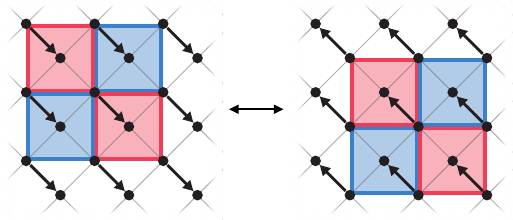}
    \caption{\textbf{The walking surface code}, implemented using late-SWAP LRUs.}
    \label{fig:walking-sc}
\end{figure}

\section{P6 hFC detectors} \label{app:detectors}

In this work, we study the P6 hFC, benchmarked in~\cite{Gidney_2021}. The basics of the P6 hFC were introduced in \cref{sec:prelims}. Here, we will show exactly how the edge measurements at each round are used to perform error correction, starting with a review of the procedure for the standard hFC. We then describe how this procedure is adapted for the swirling and sliding hFC variants introduced in this work.

We begin by reviewing how the stabilizers of the P6 honeycomb Floquet code are inferred from the measurement sequence. 
Recall \cref{eqn:isg}.
Error correction is performed by tracking the outcomes of the plaquette operators therein, each obtained from the product of the six edge measurement outcomes around the corresponding plaquette.

For example, consider a $Z^{\otimes 6}_{\mathbf b}$ plaquette. Its value is inferred by multiplying the outcomes of the $XX_{\mathbf r}$ and $YY_{\mathbf g}$ measurements on the six edges surrounding the plaquette, using the measurement outcomes from rounds $0\bmod3$ and $1\bmod3$, respectively. Thus, the $Z^{\otimes 6}_{\mathbf b}$ plaquettes can be inferred after round $1\bmod3$ using $XX_{\mathbf r}$ and $YY_{\mathbf g}$ check measurement outcomes from the immediately preceding two rounds. Similarly, the $X^{\otimes 6}_{\mathbf b}$ plaquettes are inferred at round $2\bmod3$, and $Y^{\otimes 6}_{\mathbf b}$ plaquettes at round $0\bmod3$.

Changes in the inferred plaquette stabilizer outcomes define the detector events supplied to the decoder. Each detector is therefore the product of two consecutive measurements of the same plaquette stabilizer, separated in time by one QEC cycle. Since each plaquette stabilizer is itself inferred from six edge measurements, a detector consists of the product of $12$ individual edge measurement outcomes. We refer to the spacetime region associated with the set of faults that can affect a given detector as its \textit{detector cell}.

A visual description of these detector cells is given on the left side of \cref{fig:static-detector}. Each detector cell contains the data qubit locations (black circles), whose errors can flip the detector. The corresponding edge measurement sequence and the sequence of inferred plaquette stabilizers are shown on the right, with arrows indicating the positive time direction.

\begin{figure}
    \centering
    \includegraphics[width=.9\linewidth]{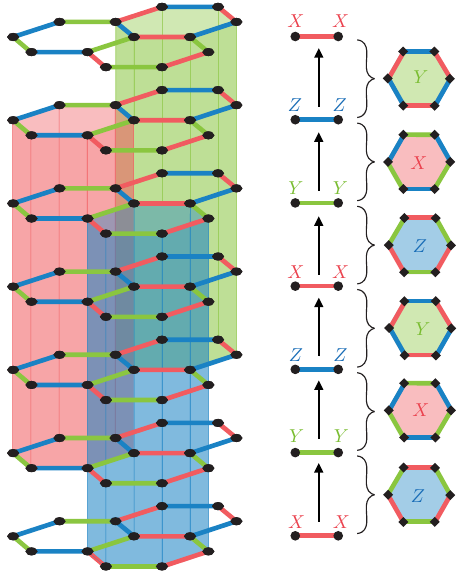}
    \caption{\textbf{Detector cells for the P6 hFC}. The right panel gives the edge measurement type associated with each timestep, with arrows pointing in the positive time direction. The corresponding plaquettes inferred at each timestep are also shown.}
    \label{fig:static-detector}
\end{figure}

In general, a $Z$-type detector is sensitive to $X$- and $Y$-type errors. Errors occurring between the two rounds whose edge measurements are used to infer the corresponding plaquette stabilizers, however, can instead be interpreted as measurement errors and need not have $X$- or $Y$-type support.

For the swirling and sliding variants, the detector construction must be modified to account for the classical updates associated with the late-SWAP LRU. Recall from \cref{fig:late-swap} that the late-SWAP circuit includes a classical update conditioned on the edge measurement outcome. This update can be interpreted as a Pauli frame change. Because the individual edge measurement outcomes are random, a late-SWAP LRU therefore introduces a randomized Pauli frame change. Consequently, the products of edge measurement outcomes used to construct the standard hFC detectors no longer directly yield deterministic detector outcomes.

To account for these frame changes, we begin with the detector cells of the standard hFC shown in \cref{fig:static-detector}. We then identify all classical updates associated with late-SWAP edge measurements adjacent to the detector cell that can flip its detector outcome. The corresponding edge measurement outcomes are incorporated into the detector. If an edge that induces such a frame update is already among the $12$ edge measurements comprising the original detector, its contribution is instead removed.

For both the swirling and sliding hFCs, the resulting detectors remain products of $12$ individual edge measurement outcomes spanning five timesteps, as in the standard hFC. However, the set of edges differs: some of the original edges are removed, while diagonal edges adjacent to the original detector cell may be included. Thus, the late-SWAP-induced Pauli frame updates modify the spacetime support of the detector without changing the total number of edge measurement outcomes entering it.

To preserve the symmetry, the classical updates induced by late-SWAP LRUs follow cyclically with the measurement basis: $XX_{\mathbf r}$ measurements induce $Y$-type updates, $YY_{\mathbf g}$ measurements induce $Z$-type updates, and $ZZ_{\mathbf b}$ measurements induce $X$-type updates, as illustrated by the late-SWAP example circuit in \cref{fig:late-swap}.

\section{Fault tolerance}\label{app:fault-tol}

In this section, we prove that the standard, late-SWAP, and early-SWAP implementations maintain the fault-tolerance of the original hFC circuit. We use the notion of fault equivalence introduced in Ref.~\cite{rodatz2026faulttoleranceconstruction}: two circuits are fault-equivalent if every undetectable fault in one circuit has a corresponding undetectable fault in the other circuit. 
For the purposes of this section, we assume the reader is familiar with the ZX diagrams and notation of Ref.~\cite{rodatz2026faulttoleranceconstruction}.

In Propositions \ref{standard-proof}, \ref{late-swap-proof} and \ref{early-swap-proof}, we show fault tolerance by a sequence of fault-equivalent rewrites. The rewrites follow from Propositions~6.9, 6.10, 6.12, and 6.13 in Ref.~\cite{rodatz2026faulttoleranceconstruction}, together with the OCM rule. In each case, we get that each circuit is fault-equivalent to the fault-tolerant $ZZ$ measurement of Proposition~6.7, and is therefore fault-tolerant. For the benefit of the reader, we reproduce the relevant propositions for fault-equivalent rewrites from Ref.~\cite{rodatz2026faulttoleranceconstruction} below:
\begin{figure}[ht]
    \centering
\includegraphics[width=\linewidth]{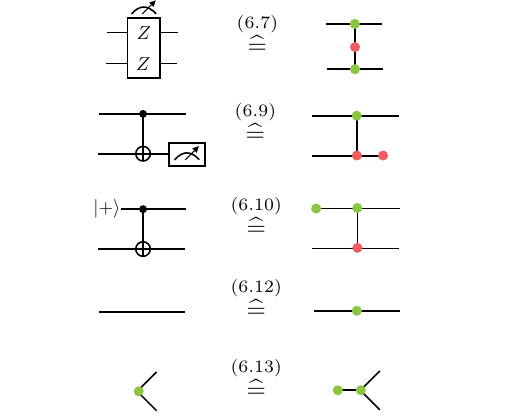}
\end{figure}

In addition to these propositions, we use the Only Connectivity Matters (OCM) rule, which enables us to arbitrarily bend the wires of ZX diagrams provided that the system inputs, outputs, and connectivity of all nodes remains unchanged.

\begin{proposition}\label{standard-proof}
The \textbf{standard} syndrome extraction circuit measuring $ZZ$ for the hFC is fault-tolerant.
\begin{figure}[ht]
    \centering
    \includegraphics[width=\linewidth]{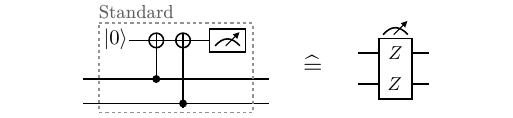}
\end{figure}
\end{proposition}

\begin{proof}
\begin{figure}[ht]
    \centering
    \includegraphics[width=\linewidth]{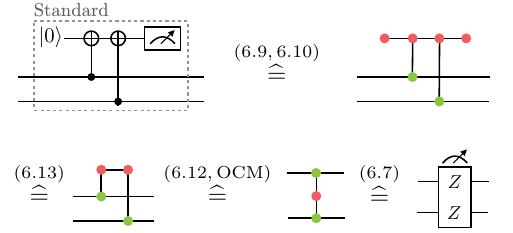}
\end{figure}
 
\end{proof}

\begin{proposition}\label{late-swap-proof}
The \textbf{late-SWAP} syndrome extraction circuit measuring $ZZ$  for the hFC is fault-tolerant.
\begin{figure}[H]
    \centering
    \includegraphics[width=\linewidth]{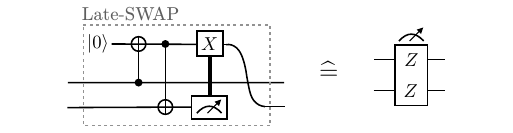}
\end{figure}
\end{proposition}

\begin{proof}
As is standard in fault-tolerant quantum computation, we assume that classical processing is error-free and does not contribute additional fault locations.
\begin{figure}[H]
    \centering
    \includegraphics[width=\linewidth]{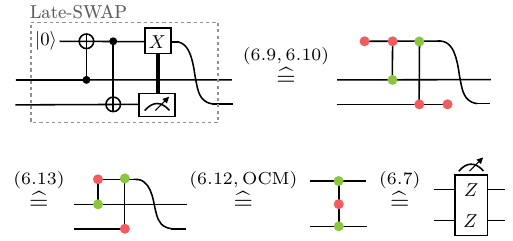}
\end{figure}
 
\end{proof}

\begin{proposition}\label{early-swap-proof}
The \textbf{early-SWAP} syndrome extraction circuit measuring $ZZ$ for the hFC is fault-tolerant.
\begin{figure}[H]
    \centering
    \includegraphics[width=\linewidth]{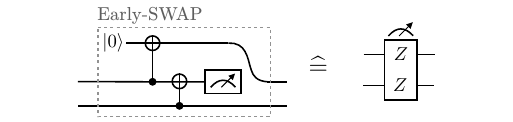}
\end{figure}
\end{proposition}

\begin{proof}
\begin{figure}[H]
    \centering
    \includegraphics[width=\linewidth]{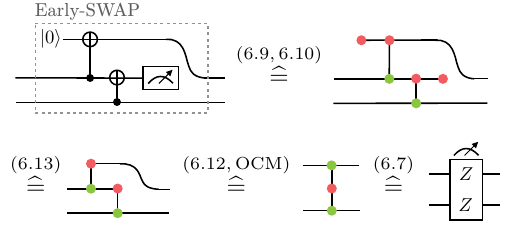}
\end{figure}
 
\end{proof}

\begin{proposition}
Edge measurements of all types ($XX$, $YY$, $ZZ$) implemented via the standard, early-SWAP, or late-SWAP circuits are fault tolerant.
\end{proposition}
\begin{proof}
From \ref{standard-proof}, \ref{late-swap-proof} and \ref{early-swap-proof}, we have that $ZZ$ measurements implemented via standard, early-SWAP, or late-SWAP circuits are fault-tolerant.
Arbitrary two-body Pauli measurements can be implemented by conjugating the measurement with single-qubit Clifford gates. Since Clifford conjugation maps Pauli faults to Pauli faults without changing their weight, this preserves fault equivalence. Thus, all edge measurements in the hFC implemented via the standard, late-SWAP, or early-SWAP circuits are fault-tolerant.
\end{proof}

\section{More LRU-hFC circuits}\label{app:allhfcs}

In \cref{sec:circs}, we introduced the two main zero-overhead leakage-reducing circuits for the hFC studied in this work: the swirling and sliding hFCs. Here, we introduce three additional walking hFC variants. First, we modify the swirling hFC to enable \textit{pipelining}, which can be advantageous on architectures where measurement operations are comparatively slow, and we do so while retaining the zero-overhead character of the swirling hFC. We then relax the zero-overhead constraint and consider two additional constructions that maintain the same circuit depth while allowing additional gates and ancilla qubits in order to be compatible with fixed qubit connectivity. The \textit{wobbling hFC} is implemented on nearest-neighbor square connectivity, while the \textit{teetering hFC} is tailored to a heavy-hex lattice.

\subsection{Pipelining the swirling hFC}
In some settings, it may be advantageous to introduce \textit{pipelining}, in which successive syndrome extraction operations are initiated before the preceding operations have completed. This is particularly natural on neutral-atom architectures, where measurement operations can be relatively slow. In this setting, one can delay all measurement operations until the end of a full QEC round (three edge measurement subrounds) and perform these measurements concurrently. This comes at the cost of additional qubit overhead: as in the swirling hFC, $5d^2$ qubits are required, compared with $3d^2$ for the sliding hFC, since an ancilla must be available for every edge of the honeycomb lattice.

It is desirable to minimize the time between successive measurements by pipelining the intervening syndrome extraction operations. Ideally, the first entangling gate of a subsequent syndrome extraction round is performed concurrently with the second entangling gate of the current round, rather than waiting for the current round to complete. Disregarding single-qubit Hadamard gates, this achieves the minimum possible separation between successive syndrome extraction operations.

To enable this pipelining in the swirling hFC while retaining the property that data qubits are reset every other round, we alternate between the late-SWAP and early-SWAP implementations. Specifically, the late-SWAP circuit is used in even rounds and the early-SWAP circuit in odd rounds. With this arrangement, the first entangling gate in every edge measurement acts on the gray qubit, while the second acts on the black qubit. Consequently, black qubits are reset every even round and gray qubits every odd round, as in the standard swirling hFC, while the syndrome extraction operations can be naturally pipelined to achieve the optimized circuit depth.

Furthermore, from \cref{tab:sliding-counting}, we see that all physical qubits experience only the sequences of roles labeled B or C. Thus, as in the standard swirling hFC, every physical qubit has a consistent lifetime of 4 gates.

\subsection{Wobbling hFC}
Thus far, we have considered only \textit{zero-overhead} approaches to the walking hFC, in which no additional gates, gate layers, or qubits are required relative to the standard hFC circuit. We now relax this constraint and consider constructions that maintain the same circuit depth while allowing additional gates and ancilla qubits during each edge measurement. This tradeoff enables shorter qubit lifetimes than are possible in the zero-overhead setting.

We first introduce the \textit{wobbling hFC}, which is implemented on nearest-neighbor square connectivity. We then introduce the \textit{teetering hFC}, which is tailored to a heavy-hex connectivity lattice.

The \textit{wobbling hFC} uses two new syndrome extraction circuits, each designed for a different local connectivity between the two data qubits being measured and the two ancilla qubits onto which their information is transferred. We call these circuits edge teleportation circuits (ETCs), because the information of both data qubits (the entire edge) is being transferred to fresh ancilla. The first, \textit{ETC-A}, shown in \cref{fig:app-late-swap-teleport}, is identical to the late-SWAP SEC, except that the data qubit not involved in the late-SWAP is teleported to a freshly initialized ancilla. As in the late-SWAP circuit of \cref{fig:late-swap}, these SWAPs can be compiled into the circuit without additional gate depth. The resulting circuit simultaneously performs the edge measurement and transfers the information from both data qubits onto freshly initialized ancilla qubits. Relative to standard syndrome extraction, ETC-A requires one additional ancilla qubit and one additional two-qubit gate, but no additional gate layers.

\begin{figure}[ht]
\centering
\includegraphics[width=\columnwidth]{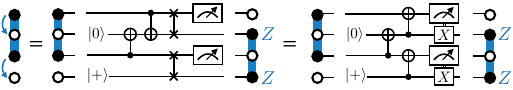}
\caption{\textbf{Edge teleportation circuit A.} We show the circuit implementation of ETC-A. The first equality is a schematic representation of the circuit, where just before measurement a SWAP is applied between each ancilla/data pair. The right-hand side employs~\cref{fig:identities}(a) and (c), yielding the final circuit. In general, ETC-A exchanges the roles of data and ancilla, with the solid, colored arrow indicating the direction of the data qubit information transfer.}
\label{fig:app-late-swap-teleport}
\end{figure}

The wobbling hFC also uses a second syndrome extraction circuit, which we call \textit{ETC-B}. This circuit is constructed from a Shor-style syndrome extraction circuit in which the two ancilla qubits are initialized in a cat state. After entangling the ancillae with the data qubits, a SWAP is applied between each data-ancilla pair immediately before measurement. Thus, the information from both data qubits is transferred onto freshly initialized ancilla qubits. The circuit implementing this procedure is shown in \cref{fig:app-shor-circuit-late-swap}.

\begin{figure}[ht]
\centering
\includegraphics[width=\columnwidth]{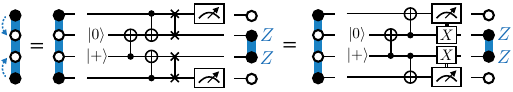}
\caption{\textbf{Edge teleportation circuit B.} We show the circuit implementation of ETC-B. The first equality is a schematic representation of the circuit, where just before measurement a SWAP is applied between each ancilla/data pair. The right-hand side employs~\cref{fig:identities}(a) and (c), yielding the final circuit. In general, ETC-B exchanges the roles of data and ancilla, with the dashed, colored arrow indicating the direction of the data qubit information transfer.}
\label{fig:app-shor-circuit-late-swap}
\end{figure}

\begin{figure*}[t]
    \centering
    \includegraphics[width=.8\textwidth]{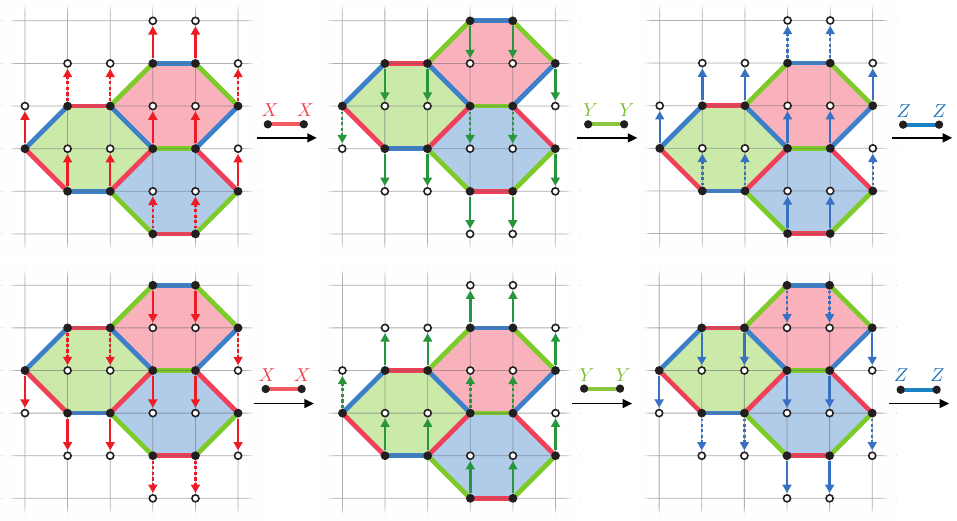}
    \caption{\textbf{Schedule of edge teleportation circuits A and B in the wobbling hFC}, which is compatible with fixed-qubit architectures with square-connectivity. In the patch, solid colored arrows indicate ETC-A, while dashed colored arrows indicate ETC-B. The choice of ETC is assigned according to the available local connectivity of data and ancilla qubits, with arrows indicating the direction of data qubit information transfer.
    Note that in any given panel, colored edges only represent the support of  $PP_\mathbf{c}$ operators, i.e., the square connectivity (gray grid) does not support connections between the qubits at the endpoints of diagonal edges.}
    \label{fig:wobbling}
\end{figure*}

As illustrated by the circuit diagrams in \cref{fig:app-shor-circuit-late-swap,fig:app-late-swap-teleport}, the two circuits require different local connectivity between the data and ancilla qubits. This allows the circuits to be strategically arranged on a nearest-neighbor square architecture. The resulting construction is shown in \cref{fig:wobbling}. At each edge measurement subround, either ETC-A or -B is used on each edge, depending on the available local connectivity of the data and ancilla qubits. Because both circuits transfer the information from both data qubits to freshly initialized ancillae, all data qubit information is transferred at every subround. The resulting patch therefore ``wobbles'' back and forth between the two lattice configurations from one subround to the next.

\subsection{Teetering hFC}

\begin{figure*}[t]
    \centering
    \includegraphics[width=.8\textwidth]{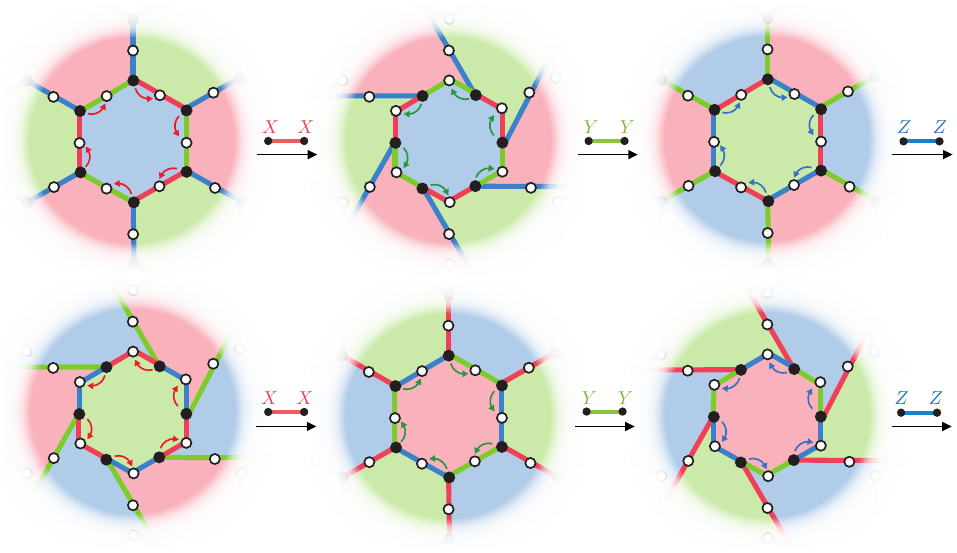}
    \caption{\textbf{Schedule of edge teleportation circuits in the teetering hFC}, which is compatible with fixed-qubit architectures with heavy-hex connectivity. In the patch, solid colored arrows indicate ETC-A, with arrows indicating the direction of data qubit information transfer.}
    \label{fig:teetering}
\end{figure*}

The second construction in this relaxed setting is the \textit{teetering hFC}, which likewise maintains the circuit depth of the standard hFC while allowing additional gates and ancilla qubits. Unlike the wobbling hFC, the teetering hFC uses only ETC-A and is tailored to a fixed heavy-hex nearest-neighbor connectivity.

In each subround of the teetering hFC, all data qubit information is transferred onto freshly initialized ancilla qubits. This is achieved on the heavy-hex layout through a sequence of alternating transfers around each plaquette. At round $0\bmod6$, we consider the blue plaquettes, whose edges are red and green. We first measure $XX_{\mathbf r}$ using ETC-A, transferring the data-qubit information clockwise to neighboring ancilla qubits around each plaquette. In the following round, $1\bmod6$, we measure $YY_{\mathbf g}$ using ETC-A and simultaneously transfer the data qubit information counterclockwise, returning it to its original lattice locations.

At this point, the blue plaquettes have ``teetered'' from one configuration to the other and back, restoring the original lattice. We then repeat the same procedure for the green plaquettes, measuring their blue and red edges in the subsequent two measurement rounds. The procedure continues cyclically, ultimately giving a period of six. The complete sequence is summarized in \cref{fig:teetering}.

\section{Detailed counting of qubit lifetimes}\label{app:counting}

In \cref{sec:qublifetimes}, we counted the qubit lifetime of every physical qubit in the swirling hFC. Here, we perform the same analysis for the sliding hFC, where the calculation is somewhat more involved. We begin by labeling the roles played by physical qubits in the late-SWAP and early-SWAP circuits, shown in \cref{fig:sliding-counting}. In each round of the sliding hFC, every physical qubit plays a role in either the late-SWAP or early-SWAP circuit. More precisely, following a reset, each qubit first serves as an ancilla ($a_l$ or $a_e$ in \cref{fig:sliding-counting}). In the following round, it becomes the data qubit that is not involved in the implicit SWAP ($d_0$), and in the subsequent round it becomes the data qubit that is swapped with the ancilla ($d_l$ or $d_e$).

\begin{figure}
\centering
\includegraphics[width=\columnwidth]{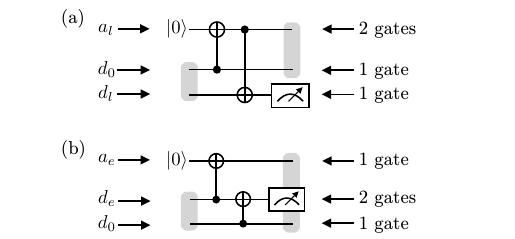}
\caption{\textbf{Counting the qubit lifetime for the sliding hFC.}
(a) Late-SWAP circuit for a $ZZ$-edge measurement. Data qubits are labeled $d_0$ and $d_l$, with the ancilla $a_l$ swapping with the $d_l$ qubit during the circuit. On the right hand side are the gate counts for each role that a physical qubit may play in the circuit ($a_l$, $d_0$, $d_l$). (b) shows the same for the early-SWAP circuit.}
\label{fig:sliding-counting}
\end{figure}

As shown on the right-hand side of \cref{fig:sliding-counting}, the number of gates experienced by a physical qubit depends on the role it plays in each circuit. Ancilla qubits experience either $2$ gates in the late-SWAP circuit ($a_l$) or $1$ gate in the early-SWAP circuit ($a_e$). The $d_0$ role always involves $1$ gate, while the $d_l$ and $d_e$ roles involve $1$ and $2$ gates, respectively. Thus, a physical qubit can follow one of four possible sequences of roles from reset to measurement, resulting in lifetimes of $3$, $4$, or $5$ gates, as summarized in \cref{tab:sliding-counting}.

\begin{table}[t]
    \centering
    \caption{\textbf{Physical qubit lifetimes in the sliding hFC.}}
    \label{tab:sliding-counting}

    \begin{tabular*}{.8\columnwidth}{@{\extracolsep{\fill}}ccc}
        \toprule
        \textbf{Sequence of Roles}
        & \textbf{Qubit Lifetime}
        & \textbf{Label} \\
        \midrule
        $a_e-d_0-d_l$
        & 3
        & A \\
        $a_l-d_0-d_l$
        & 4
        & B \\
        $a_e-d_0-d_e$
        & 4
        & C \\
        $a_l-d_0-d_e$
        & 5
        & D \\
        \bottomrule
    \end{tabular*}
\end{table}

On a torus, the sliding hFC returns to its original configuration every $6$ measurement rounds. During these $6$ rounds, each physical qubit is reset twice, after which the process repeats. Tracking each physical qubit over this period reveals two distinct classes of trajectories. Two thirds of the physical qubits alternate between lifetimes of $3$ and $4$ gates, following either the $A$-$B$ or $A$-$C$ trajectory in \cref{tab:sliding-counting}. The remaining one third consistently follow trajectory $D$, giving them a fixed lifetime of $5$ gates. \cref{fig:sliding-patch-counting} summarizes the lifetime of each physical qubit according to its relative location in the initial patch shown in \cref{fig:sliding}. In particular, every third row of physical qubits has a lifetime of $5$, while the remaining qubits alternate between lifetimes of $3$ and $4$ gates.

\begin{figure}
\centering
\includegraphics[width=\columnwidth]{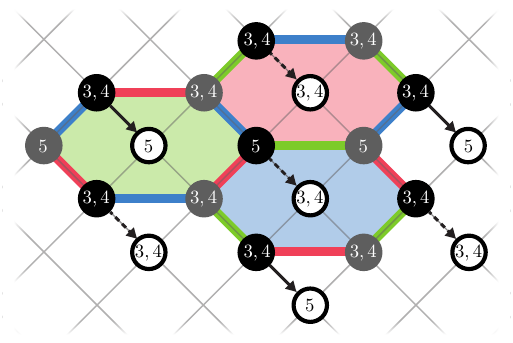}
\caption{\textbf{Qubit lifetimes across the sliding hFC patch.} Qubits with a consistent lifetime of $5$ always play the sequence of roles labeled D in \cref{tab:sliding-counting}. All other qubits alternate between a lifetime of $3$ or $4$ gates. That is, they either exhibit an alternating sequence of A-B or A-C from \cref{tab:sliding-counting}.}
\label{fig:sliding-patch-counting}
\end{figure}

\begin{figure*}[t]
    \centering
    \includegraphics[width=.9\textwidth]{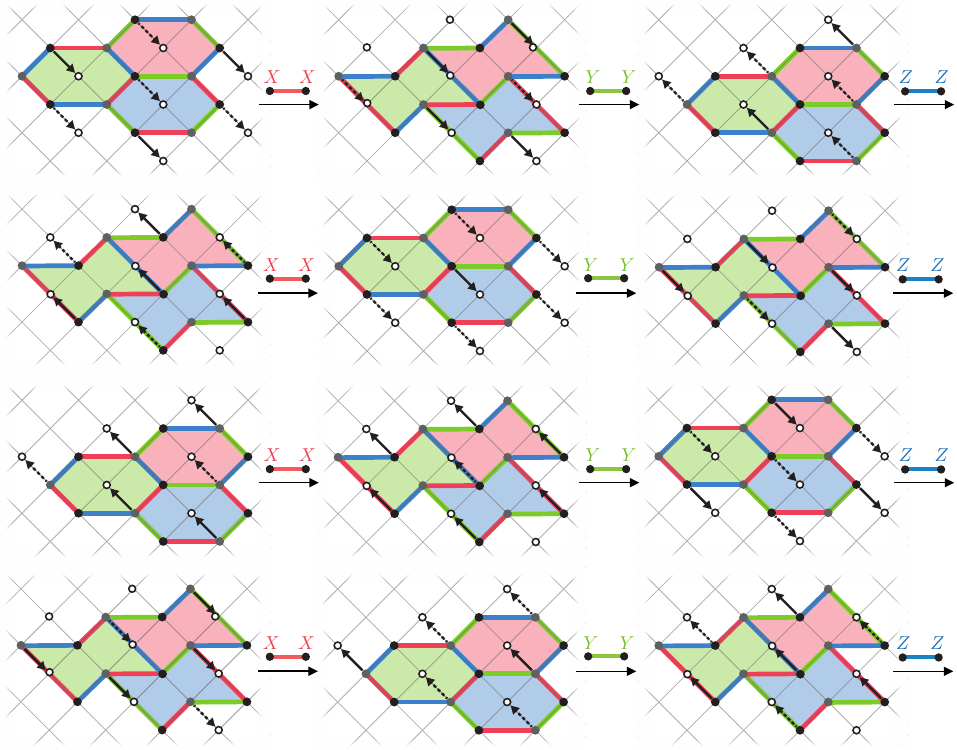}

    \caption{\textbf{Schedule of early- and late-SWAP LRUs in the sliding hFC adapted to a planar layout}, which is compatible with fixed-qubit planar architectures with square-connectivity. In the patch, solid arrows indicate late-SWAP LRUs, while dashed arrows indicate early-SWAP LRUs. The choice of LRU circuit is assigned according to the available local connectivity of data and ancilla qubits, with arrows indicating the direction of data qubit motion. Unlike the sliding hFC implemented on a torus, black and gray qubits are \textit{not} necessarily reset on even and odd rounds respectively.}

    \label{fig:sliding-wiggle}
\end{figure*}

In order to adapt the sliding hFC to a planar layout, the patch must ``wiggle'' back and forth rather than slide continuously in one direction. This motion, shown in \cref{fig:sliding-wiggle}, has an overall period of $12$. Unlike the toric sliding hFC, the planar construction does not reset the black and gray qubits on alternating rounds. Instead, black qubits are reset twice consecutively, followed by two consecutive resets of the gray qubits, and so on. This produces a larger variation in qubit lifetime: physical qubits can experience between $2$ and $7$ gates between reset and measurement, while the average lifetime remains $4$ gates.

\section{Details of numerical simulations} \label{app-sims}

The simulations of the swirling and sliding hFCs, as well as the walking surface code, are implemented in Python using the \texttt{Stim}~\cite{Gidney_2021} and \texttt{PyMatching}~\cite{Higgott_2025} packages, as well as a custom wrapper for leakage noise. The circuits are constructed in \texttt{Stim}, with all two-qubit gates decomposed into CZ gates and single-qubit Hadamard gates. Example circuits are provided in \url{https://github.com/hanna-westerheim/walking-floquet-codes}. We generate a decomposed detector error model (DEM) for each circuit and distance under uniform depolarizing noise, which is then passed to \texttt{PyMatching} for minimum-weight perfect matching.

To generate syndrome data for the leakage noise models, we use a custom \texttt{Stim} wrapper that probabilistically introduces leakage events at each spacetime location and modifies the noiseless circuit to incorporate the effects of each leakage event according to the specified noise model. For example, under the depolarizing leakage model, when a qubit becomes leaked, complete depolarizing noise is applied to every qubit with which it subsequently interacts, up until the leaked qubit is reset. Under the skip-gate leakage model, all subsequent gates involving a leaked qubit are removed from the circuit. When a leaked qubit is measured, its measurement outcome is replaced by a uniformly random outcome.

For each physical error rate, we determine the required number of shots using a pilot sampling procedure. Sampling proceeds until either ten logical failures have been observed or a maximum of $10^6$ shots has been reached. We then perform ten times as many shots as were required during the pilot sampling (up to the maximum number of shots), resulting in an expected total of approximately $100$ logical failures. The statistical uncertainty in the logical error rate is estimated as
\begin{equation}
\sigma_{p_L} \approx \frac{\sqrt{N_{\mathrm{fail}}}}{N_{\mathrm{shots}}},
\end{equation}
where $N_{\mathrm{fail}}$ is the number of logical failures and $N_{\mathrm{shots}}$ is the total number of shots. This procedure is used to generate the data shown in \cref{fig:bars}(a).

For \cref{fig:bars}(b), we extract the scaling of the logical error rate by performing a linear regression on the log-log data shown in \cref{fig:bars}(a). The fit includes data points up to $p=10^{-3}$, with the $p=10^{-3}$ point included for the walking hFC variants but excluded for the walking surface code. The uncertainty in the fitted scaling exponent is obtained from the corresponding regression.

\end{document}